\documentclass[12pt]{article}
\usepackage{authblk}
\usepackage{geometry,color}
\usepackage{mathrsfs,amssymb,amstext,amsmath,amsthm,eucal}
\usepackage{hyperref}
\usepackage{tikz}\usetikzlibrary{arrows,shapes}
\newtheorem{theorem}{Theorem}[section]
\newtheorem{lemma}[theorem]{Lemma}
\newtheorem{proposition}[theorem]{Proposition}

\newcommand{\dvol}{\mathrm{dvol}}
\newcommand{\half}{\frac{1}{2}}
\newcommand{\expab}[2]{ e^{(#1\alpha #2)\phi} }

\title{An Exact Conformal Symmetry Ansatz on Kaluza-Klein Reduced TMG}

\author{George Moutsopoulos\footnote{gmoutso@googlemail.com}}
\affil{Department of Electrophysics, National Chiao-Tung University, Hsinchu, Taiwan}
\author{Patricia Ritter\footnote{p.d.ritter@sms.ed.ac.uk}}
\affil{{Maxwell Institute and School of Mathematics,
University of Edinburgh, UK}}
\begin{document}

\maketitle
\begin{abstract}
Using a Kaluza-Klein dimensional reduction, and further imposing a conformal Killing symmetry on the reduced metric generated by the dilaton, we show an Ansatz that yields many of the known stationary axisymmetric solutions to TMG.
\end{abstract}
\newpage
\section{Introduction}
Topologically massive gravity (TMG)~\cite{Deser:1981wh,Deser:1982vy} and its solutions have been studied extensively, even more so recently, due to the conjecture of~\cite{Anninos:2008fx} regarding a CFT dual to spacelike warped AdS$_3$ black holes. For a critical value of the theory, the holography of null warped AdS$_3$ was studied extensively in~\cite{Guica:2010sw}, see also \cite{Anninos:2010pm}. However, spacelike warped AdS$_3$ has different asymptotics to AdS$_3$ or the Schr\"odinger background and a similar analysis cannot be made.

The largest class of known solutions to TMG is the Kundt class~\cite{Chow:2009vt}, which includes the TMG wave~\cite{Gibbons:2008vi} and spacelike warped AdS$_3$; the odd one out is timelike warped AdS$_3$, which is not a Kundt spacetime~\cite{Anninos:2008fx}. Various other solutions can be written up to identifications with one of the above~\cite{Chow:2009km}. One of our motivations here was the search for an ``intermediate'', or ``interpolating'' solution between AdS$_3$ and spacelike warped AdS$_3$, for generic values of the theory, which could be relevant to the warped-AdS/CFT correspondence. 

Numerical solutions that are asymptotic to warped AdS$_3$ were found in \cite{Ertl:2010dh}, wherein the same question as ours is posed. The ambition was further encouraged by \cite{Guralnik:2003we}, where an interesting solution was found for the purely gravitational Chern-Simons term that appears in the TMG action. These solutions can be related at a local level to kinks with interpolating behaviour, see also~\cite{Grumiller:2003ad}. The hope was to generalise their approach to include the Einstein-Hilbert action, and search for a similar solution for the full model.

In \cite{Guralnik:2003we}, the authors used a Kaluza-Klein (KK) dimensional reduction on the three-dimensional theory, to obtain a system of differential equations in 2 dimensions. For the purely Chern-Simons part of the action, it turns out that one of the equations of motion is a conformal Killing equation on the gradient of one of the reduced fields. It is the presence of this new symmetry that allows a simple solution to the problem. 

We will see below that the approach of \cite{Guralnik:2003we} does not generalise in a simple way for the full TMG action. Recalling the classification of Pope \textit{et al.} \cite{Chow:2009km} and the Kundt solutions to topologically massive gravity \cite{Chow:2009vt}, we will show that our ``kinky'' approach only leads to a subset of these. The symmetries imposed by the Ansatz, \textit{i.e.} an isometry along which to perform the KK-reduction and an exact conformal Killing symmetry generated by the dilaton, are too restrictive to yield new solutions. The approach does however yield locally most of the known stationary axisymmetric solutions of TMG as collected in \cite{Ertl:2010dh}. 

Although these solutions are not gravitational kinks, we have retained use of the word since our method is influenced by~\cite{Guralnik:2003we}. In section 2 we set up our notation and introduce some helpful theorems to streamline our derivation. In section 3 we motivate our Ansatz and in section 4 we identify the solutions it yields. We end with concluding remarks. We also provide four complementary appendices. Appendix~\ref{app:cotton} shows the reduction of the Cotton tensor explicitly, appendix~\ref{app:null} comments on the ``null $d\psi$'' case, appendix~\ref{app:mapsol} has a map of all known solutions to TMG, and appendix~\ref{app:stataxis} comments on the stationary axisymmetric solutions with respect to our method.

\section{Setup and notation}
In this first section we derive the equations of motion of the reduced action and set up some theorems that simplify the ensuing analysis. Note that the one-dimensional reduction in $d=3$ that we perform here is in a sense the equivalent of an $S^2$ (Pauli) reduction in $d=4$.

\subsection{2d reduced action}
We write the full TMG action as
\[ 16\pi G\,S[g]=\int d^3x \sqrt{-g}\left(R + \frac{2}{\ell^2} 
 + \frac{\ell}{6\nu} \epsilon^{\lambda \mu \nu} \Gamma^r_{\lambda\sigma}
\left( \partial_\mu \Gamma^\sigma_{r\nu} + \frac{2}{3}\Gamma^\sigma_{\mu\tau}\Gamma^\tau_{\nu\rho}\right)\right)~.
\]
We follow the usual KK-reduction set-up, starting with a 3-dimensional
metric
\begin{equation*}
  \label{eq:2}
    g^{(3)}
  =e^{2\alpha\phi}\bar{g} \pm e^{2 \phi}\left(dz+ A\right)^2\,,
\end{equation*}
where we assumed the isometry $z\mapsto z+\xi$. The field $\phi$ is a function and $ A$ a one-form on the remaining two coordinates. We raise/lower the 2-dimensional tensorial indices $a$ and $b$ with the metric $\bar
g_{ab}$. The $\pm$ sign distinguishes spacelike and timelike reductions.
 We could absorb the $\alpha$ parameter above into $\bar{g}$, but
we choose to leave this free for now. This freedom will allow us to find various solutions from
one simple Ansatz. 

By $D_a$ we denote the 2-dimensional covariant derivative and write $D^2=D_aD^a$. The field strength $F=dA=f\dvol_{\bar{g}}$ defines the scalar $f$ in 2 dimensions by its Hodge dual. The 3-dimensional scalar curvature $R$ written in terms of the
2-dimensional curvature $\bar R$ is given by
\begin{equation}\label{eq:Ricciscalarred}\begin{aligned}
 R& = e^{-2\alpha\phi} \bar{R} - 2 \left( 
 \alpha+ {1}\right) e^{-2\alpha\phi} D^2\phi -2  e^{-2\alpha\phi} | d\phi|^2 + \frac{1}{2} e^{-4\alpha\phi+2\phi} f^2  ~
\end{aligned}
\end{equation}
and the Einstein-Hilbert part of the action is therefore
\begin{equation}
I_{EH}=\int \dvol_{\bar{g}} \left[ e^{\phi} \bar{R}
+  2\alpha e^{\phi} |d\phi|^2 + \frac{1}{2} e^{(-2\alpha+3)\phi } f^2  -2 (1+\alpha) D_a \left( e^{\phi}D^a\phi \right)\right] .
\end{equation}

To KK-reduce the Chern-Simons-like terms in the action, we make use of
the results of \cite{Guralnik:2003we}. Schematically, in their set-up ($\alpha={1}$)
\begin{equation*}
\left(\epsilon \Gamma (\frac{1}{2}\partial \Gamma + \frac{1}{3}\Gamma\Gamma)\right)\longrightarrow -\frac{1}{2}\sqrt{- g}(F\bar R + F^3)\,.
\end{equation*}
This schematic result has to be corrected by exponential factors for generic $\alpha$. This can be easily done, since the metric of \cite{Guralnik:2003we} is  conformally related to our generic one by $g_{ab}=e^{2(\alpha-1)\phi}\bar g_{ab}$. We thus obtain the action
\begin{equation}
I_{CS}=\pm\frac{1}{2\mu}\int \dvol_{\bar{g}} \left( e^{(-2\alpha+2)\phi} f \bar{R} -2 (\alpha-{1}) e^{(-2\alpha+2)\phi} f D^2 \phi + e^{(-4\alpha+4)\phi} f^3 \right)~,
\end{equation}
where $\mu=3\nu/\ell$. Both parts of the action are valid for either sign of the reduction. 

\subsection{Equations of motion}\label{subsec:eoms}
We can now either vary\footnote{the variation of $f$ is $\delta f = \frac{1}{2} f g_{\mu\nu} \delta [g^{\mu\nu}] \mp \epsilon^{\mu\nu}\partial_\mu \delta A_\nu $.} the reduced action, or reduce the 3-dimensional equations. Either way one obtains the same set of equations of motion. 
The consistency of dimensional reduction on a unimodular or one-dimensional group is well-known\footnote{see e.g.~\cite{Cvetic:2003jy,Pons:2006vz,Torre:2010xa}.} and still holds for the Chern-Simons correction. Never the less, we have laboriously gone both ways, and appendix~\ref{app:cotton} contains the explicit form of the Cotton tensor under the reduction. The equations of motion are:
\begin{equation}
\begin{array}{ll}
\mathbf{F:}\qquad
c&=
\pm 2 \mu\, \expab{-2}{+3} f + e^{2({1}-\alpha)\phi} \bar{R} - 2(\alpha-{1})e^{2({1}-\alpha)\phi}D^2\phi+3 e^{4({1}-\alpha)\phi}f^2,\\ \ \\
\textbf{Dil:} \,
-\frac{6}{\ell^2} &=e^{-2\alpha \phi}\bar{R}-2(\alpha+{1}) e^{-2\alpha\phi} D^2 \phi -2  e^{-2\alpha\phi} |d\phi|^2+\half \expab{-4}{+2} f^2, \\ \ \\
\textbf{Kink:}\ \ 
0&=D^2 e^{\phi} + \half \expab{-2}{+3} f^2-\frac{2}{\ell^2} \expab{2}{+1}\pm\frac{1}{2\mu}\Big( D^2(\expab{-2}{+2}f) \\ \ \\
&\ \ +\expab{-2}{+2}f \bar{R} -2(\alpha-{1})\expab{-2}{+2} f D^2 \phi +2\expab{-4}{+4} f^3 \Big), \\ \ \\ 
\textbf{CKV:}\ \ 
0 &=e^{2\alpha\phi} D_{(a}\left( e^{-2\alpha\phi} D_{b)} e^{\phi}\right)\pm\frac{1}{2\mu}e^{2(\alpha-{1})\phi} D_{(a}\left( e^{-2(\alpha-{1})\phi} D_{b)}(e^{-2(\alpha-{1})\phi} f)\right) \\ \ \\
 \ &\ \  -\bar g_{ab}\left[e^{2\alpha\phi} D^c\left( e^{-2\alpha\phi} D_c e^{\phi}\right)\pm\frac{1}{2\mu}e^{2(\alpha-{1})\phi} D^c\left( e^{-2(\alpha-{1})\phi} D_c(e^{-2(\alpha-{1})\phi} f)\right)\right].
\\ \ \\
\end{array}
\end{equation}
In the \textbf{F} equation $c$ is a constant of integration that a solution will fix. The \textbf{CKV} equation is the traceless part of the Einstein equation and round brackets around indices indicate symmetrization of strength one. The trace of the Einstein equation is what we call the \textbf{Kink} equation. For the dilaton equation (\textbf{Dil.}) we used \eqref{eq:Ricciscalarred} and the constant scalar curvature $-6/\ell^2$ of the 3d geometry. The two-dimensional equations have been examined before, e.g. in the conformal gauge in \cite{Clement:1990mp}.

The equations exhibit two types of ``symmetry'': a scaling of $z\mapsto \xi\,z$ and a shift of $\alpha\mapsto \alpha + \tilde\xi$. The former rescales the fields as $e^{\phi} \mapsto \xi^2 e^{\phi}$, $\bar{g} \mapsto \xi^{-2\alpha}\bar{g}$ and $f\mapsto \xi^{2\alpha-1}f$, and can be used to normalize $c$. The latter transforms fields as $\bar{g} \mapsto e^{-2\tilde\xi\phi}\bar{g}$ and $
f \mapsto e^{2\tilde\xi\phi} f$~, 
whereas it leaves $\phi$ unchanged. Using this, $\alpha$ can be fixed from the onset but, as mentioned above, we keep this freedom and let it be fixed by a consistency requirement on our Ansatz.

\subsection{Conformal Killing vectors}\label{sec:theorems}
Let us now focus on the last of the above equations of motion, labeled \textbf{CKV} because of its similarity to a conformal Killing equation. In fact it contains the conformal Killing equation of \cite{Guralnik:2003we} for $D_a (e^{-2(\alpha-{1})\phi}f)$, coming from the purely Chern-Simons part of the action, but is complicated by the similar equation for $D_a e^{\phi}$ coming from the Einstein-Hilbert term. Nevertheless, this equation motivates us to search for solutions where its content is that of a {single} conformal Killing vector equation. 

This is both for simplicity, but also in the hope of finding  behaviour similar to \cite{Guralnik:2003we}. Let us first list a set of propositions that will help us in the subsequent analysis. The first proposition is used to fix the metric.

\begin{proposition}\label{prop:dpsi}

If $\bar{g}$ has a conformal Killing one-form $d\psi$ that is non-null and exact, then the metric can be written in some coordinate system as
\begin{equation}
\label{eq:dpsimetric} 
\bar{g} = \frac{d\psi(x)}{dx}(dx^2-dt^2) ~.
\end{equation}
\end{proposition}
\begin{proof}
In a conformal gauge, the metric can be written as
$ g = \Lambda(x,t) du dv\,, $
where conformal Killing vectors are of the form
$ X = g(v) \partial_v + h(u)\partial_u ~.$
The condition that $ \bar{g}(X)$ is non null ($g(v) h(u)\neq0$), allows us to change coordinates
\begin{align*}
u &\mapsto \int \frac{1}{h(u)}, \\
v &\mapsto \int \frac{1}{g(v)},
\end{align*}
so that $g = \Lambda(x,t) (dx^2-dt^2)$ with $X=\partial_x$, and $\bar{g}(X)=d\psi$ implies $\Lambda=\psi'(x)$.
\end{proof}

The following proposition will be needed to complete our Ansatz.
\begin{proposition}\label{prop:twopsi}
Assume two non-null conformal Killing one-forms, $F_1 d F_2$ and $d\psi$, with $F_1$ and $F_2$ functions of $x$ in the adaptive coordinate system of proposition \ref{prop:dpsi}, equation \eqref{eq:dpsimetric}. They are necessarily related by
\[ F_1\, d F_2 = \tilde{k}\, d\psi \]
for some constant $\tilde{k}$.
\end{proposition}
\begin{proof}
$F_1 d F_2$ is dual to a conformal Killing vector $X = g(v) \partial_v + h(u)\partial_u ~$,
for some functions $g(v)$ and $h(u)$. Since the left hand side of
\[ F_1 d F_2 = \frac{\psi'(x)}{2}\left( \left(g(x+t) + h(x-t)\right)dx +  \left(g(x+t) - h(x-t)\right)dt \right) ~.\]
is a function of $x$, we have $g(x+t)=h(x-t)=\text{const.}~.$
\end{proof}
Finally, we have
\begin{proposition}\label{prop:Zcons}
Take $d\psi$ to be the metric dual to a conformal Killing vector as before. Then in the adapted coordinates we define
\[ Z= \frac{1}{\psi'} \frac{d}{dx} \]
for which
\begin{equation}\label{Z and R} Z D^2 \psi = - \bar{R} ~.\end{equation}
\end{proposition}
\begin{proof}
For $g= e^{2\sigma(x)}(dx^2-dt^2)$, the Laplacian is
$D^2=e^{-2\sigma(x)}(\partial_x^2-\partial_t^2)$. At the same time, the curvature scalar is
 $\bar{R}= -2 e^{-2\sigma(x)}\sigma''(x)$. We substitute $\psi'(x)=e^{2\sigma(x)}$.\end{proof}

The reason why the purely Chern-Simons term in~\cite{Guralnik:2003we} has a unique solution up to homothety, is precisely because the Chern-Simons term is conformally invariant and, with $\alpha=1$, the \textbf{CKV} equation becomes that of a single exact conformal Killing vector equation. Proposition~\ref{prop:dpsi} is used here as in \cite{Guralnik:2003we} to fix the metric. However, the \textbf{CKV} equation here is at best a sum of two such ``exact conformal Killing vector equation'' terms. If we impose that they vanish separately, propositions \ref{prop:dpsi} and \ref{prop:twopsi} combined give a class of unique candidate solutions. Proposition~\ref{prop:Zcons} is then used as in \cite{Guralnik:2003we} to check for the consistency of the candidate solution. 

Equivalent statements to the above propositions for a \emph{null} one-form $d\psi$ can also be written. However, our method for the null case does not lead to any solutions. We comment on the null case in appendix~\ref{app:null}.

\section{A general Ansatz}\label{sec:Ansatz}
Before moving onto a general Ansatz involving functions generating conformal Killing vectors, we glance briefly at the simplest solution to the equations of motion.
\subsection{Constant \texorpdfstring{$f$ or $\phi$}{fields}}
\label{sec:example-solutions}
From our Kaluza-Klein Ansatz, it is clear that we can obtain known solutions to TMG by simply setting $f$ and $\phi$ to constant values $f=f_0$, $\phi=\phi_0$. For simplicity, let us set here $\alpha=1$. From the {dilaton} (\textbf{Dil.}) equation of motion we obtain $\bar R$ in terms of these constants, while the \textbf{Kink} equation becomes
\begin{equation} \frac{1}{2} ( e^{\phi_0} \pm \frac{3}{2\mu} f_0)(f_0 - \frac{2}{\ell}e^{\phi_0})(f_0 + \frac{2}{\ell}e^{\phi_0})=0\, ,\end{equation}
yielding AdS$_3$ or warped AdS$_3$, respectively for $f_0= \pm \frac{2}{\ell}e^{\phi_0}$ and  $e^{\phi_0} =\mp \frac{3}{2\mu} f_0$. 

Along these lines, it is interesting to note that constancy for $\phi$ implies the same for $f$, and vice versa. This can be easily checked by setting one of the two functions to a constant value and studying the equations of motion for the other. A short treatment for when $f=0$ is given in~\cite{Deser:2009er}.

\subsection{The Ansatz}\label{sec:ansatz}
Let us focus again on the \textbf{CKV} equation. If we view this as the sum of two conformal Killing equations coming separately from the Einstein part and Chern-Simons part, we can only obtain the AdS$_3$ solution. Trying to relax this idea, we can allow for a ``mixing'' of the functions appearing in the two gradients. For instance, focus on the first term
$$e^{2\alpha\phi}D_a(e^{-2\alpha\phi}D_b e^{\phi})= e^{\phi} ( D_a D_b \phi + ({1}-2\alpha)D_a\phi D_b\phi),$$
and write out the function $f$ as 
\begin{equation}f= \pm 2\mu k \expab{2}{-1} + \expab{2}{-2}\tilde{f} ~\end{equation}
for a constant $k$. Inserting this into the second term of the \textbf{CKV} equation yields
\begin{equation*}
\begin{array}{rl}
\pm\frac{1}{2\mu}\expab{2}{-2} D_{(a}\left( \expab{-2}{+2} D_{b)}(\expab{-2}{+2}f)\right)=& k \, e^{ \phi} ( D_a D_b \phi + (-2\alpha +3 )D_a\phi D_b\phi) \\ \\ 
&\pm \frac{1}{2\mu}\expab{2}{-2}D_{(a}(\expab{-2}{+2}D_{b)}\tilde{f}).
\end{array}
\end{equation*}

The most obvious approach is to impose that $\tilde{f}$ is zero so that we are left with the conformal Killing vector equation. For $k\neq\tfrac{{1}-2\alpha}{2\alpha-3}$ the left-hand side of the equation becomes
\begin{multline*}
e^{\phi} (1+k)\left(D_a D_b \phi + \frac{({1}-2\alpha)+k(-2\alpha+3)}{1+k} D_a\phi D_b \phi\right)
=  e^{({1}-\epsilon)\phi}\frac{1+k}{\epsilon}D_a D_b e^{\epsilon\phi},\end{multline*}
when
\[\epsilon = \frac{({1}-2\alpha)+k(-2\alpha+3)}{1+k} \]
is well defined and non-zero. 
That is, $k\neq -1$ and $k\neq\tfrac{{1}-2\alpha}{2\alpha-3} $. It is then natural to take $d\psi=d e^{\epsilon\phi}$ in proposition \ref{prop:dpsi}. If, on the other hand, we start by imposing $d\psi=d e^{\epsilon\phi}$, then $\tilde{f}$ appears in the \textbf{CKV} equation that now takes the form of a conformal Killing vector equation, and so is fixed by using the \textbf{F} equation of motion to satisfy proposition \ref{prop:twopsi}. This way all fields are fixed and in particular
\begin{subequations}\begin{align} 
f& = \pm2\mu k \expab{2}{-1}+\tilde{k} e^{(4\alpha-4+\epsilon)\phi}+\delta \expab{2}{-2} & \text{when }\epsilon&\neq 2-2\alpha \Leftrightarrow k\neq1\\
 f &= \pm2\mu k \expab{2}{-1}+ \expab{2}{-2}(\tilde{k}\phi+\delta)&\text{when }\epsilon&=2-2\alpha \Leftrightarrow k=1~.
\end{align}\end{subequations}

The metric one obtains by choosing the conformal Killing generator to be $\psi=e^{\epsilon\phi}$ for some $\alpha$ is equivalent to the one obtained by the choice $\psi=\phi$ for $\alpha'=\alpha+\epsilon/2$. Our Ansatz is thus to assume $d\psi=d\phi$ is a conformal Killing one-form. We set
\[k=\frac{1-2\alpha}{2\alpha-3} ~,\]
and by using proposition \ref{prop:twopsi}, which is satisfied by the \textbf{F} equation of motion,  $f$ is given by
\begin{subequations}\label{eq:finalAnsatz}\begin{align}
 f &= \pm 2\mu k \expab{2}{-1} + \tilde{k}\expab{4}{-4}+\delta\expab{2}{-2} &\textrm{if }\alpha&\neq\frac{3}{2},1 \\
 f &= \pm 2\mu  e^\phi + \tilde{k}\phi+\delta &\textrm{if }\alpha&=1~,
 \end{align}\end{subequations}
whereas the metric is given by \eqref{eq:dpsimetric} with $\psi=\phi$.  Our Ansatz can thus be summarized by the statement $\psi=\phi$ in proposition \ref{prop:dpsi}. This way the \textbf{CKV} equation is automatically satisfied and at the same time all fields are fixed. It remains to show that the other {three} equations of motion are satisfied for suitable values of $\alpha$, $\tilde{k}$, $\delta$ and $c$. In appendix \ref{app:stataxis} we compare our method to what has been done for stationary axisymmetric solutions.

\section{Solutions}\label{sec:solutions}
In this section we check the consistency of our Ansatz, namely which functions $f$ and $\phi$ related by our Ansatz satisfy the reduced TMG equations of motion. Starting with \eqref{eq:finalAnsatz}, we use the equations of section \ref{subsec:eoms} to calculate the expressions for $|d\phi|^2$, $\bar R$ and $D^2\phi$ in terms of $\phi$. We then use proposition \ref{prop:Zcons} and compare $ZD^2\phi$, that is $Z$ acting on the expression for $D^2\phi$, with the expression for $-\bar R$ obtained previously. When the two expressions match, the equation for $D^2\phi$ implies that of $\bar R$. Finally, the consistency of the equation for $|d\phi|^2=\phi'$ is checked by the integral of the equation for $D^2\phi=\phi''/\phi'$. Schematically, the consistency involves checking the following derivations
\begin{align*}
|d\phi|^2 &= F(\phi;\alpha,c,\tilde{k},\delta) \\ &\Downarrow  G=\partial_\phi F \\ 
D^2\phi &= G(\phi;\alpha,c,\tilde{k},\delta)\\&\Downarrow H=\partial_\phi G\\
\bar{R} &= H(\phi;\alpha,c,\tilde{k},\delta)~.
\end{align*}

The resulting conditions are in terms of long expressions involving exponentials of $\phi$, schematically
\[ \sum_{(m,n)\in S} \expab{m}{+n}~.\]
Recall the first consistency check is an equation of the type
\begin{equation*} Z D^2\phi +\bar{R} = 0 ~.\end{equation*}
The simplest approach is to consider all the powers to be different, $m\alpha+n\neq m'\alpha+n'$, so that their coefficients have to vanish separately. We thus obtain three cases:
\begin{enumerate}\label{eq:3cases}
\item $\delta=0$, $\alpha={1}/2$ and $\tilde{k},c$ unconstrained;
\item $c-\delta^2=\tilde{k}=\alpha=0$;
\item $\delta=c=\tilde{k}=0$ and $\mu^2\ell^2(2\alpha+{1})^2=(2\alpha-3)^2~.$
\end{enumerate}
One need also check the cases when the powers mentioned above are not all different. This happens when 
\[\alpha=0,1/2,3/4,1,9/8,7/6,5/4,4/3,11/8,5/3,7/4,2,5/2,3.\]
For each of these values we simplify the result, but again find the same three possible solutions.

The final check is to verify that the expression for $|d\phi|^2$ is also satisfied. We therefore integrate the expression for $D^2\phi$
\[ D^2\phi=H(\phi)\rightarrow \phi''=H(\phi)\phi' \rightarrow |d\phi|^2=\phi'=\int H d\phi + d ~,\]
for a function $H(\phi)$ of $\phi$, and compare with the expression for $|d\phi|^2$. One finds that for a suitable integration constant $d$, the two expressions always match for the three cases above. 

We will now write down and identify the three classes of solutions that can be obtained via our Ansatz.

\subsection{Case 1\texorpdfstring{: $\delta=0$, $\alpha={1}/2$}{}}
In this case our generalised Ansatz simply becomes
\[ f = \tilde{k} e^{-2\phi} ~,\]
so that $  A = -\frac{\tilde{k}}{2}e^{-2\phi} dt$.
 Solving the equations of motion we get
\begin{equation*}
\begin{array}{ll}
D^2\phi &=\frac{1}{2} ( c \mp 2\mu\tilde{k})e^{-\phi}+\frac{1}{\ell^2}e^\phi-\frac{3}{4}\tilde{k}^2 e^{-3\phi}, \\ \\
\bar{R} &= \frac{1}{2} (c\mp 2\mu\tilde{k})e^{-\phi}-\frac{1}{\ell^2} e^{\phi}-\frac{9}{4}\tilde{k}^2 e^{-3\phi}.
\end{array}
\end{equation*}
Integrating $D^2\phi=\phi''/\phi'$ 
 \[ |d\phi|^2=\phi'=-\frac{1}{2} ( c \mp 2\mu\tilde{k})e^{-\phi}+\frac{1}{\ell^2}e^\phi+\frac{1}{4}\tilde{k}^2 e^{-3\phi} +d\]
and inserting into the dilaton equation (along with $D^2\phi$) we find that $d=0$. 

The 3-dimensional metric can now be written in the $\phi$ coordinate as
\begin{equation} g = e^\phi \left( \frac{d\phi^2}{\frac{\tilde{k}^2}{4} e^{-3\phi} + \frac{\gamma}{2}e^{-\phi}+\frac{1}{\ell^2} e^{\phi}} \mp (\frac{\tilde{k}^2}{4} e^{-3\phi} + \frac{\gamma}{2}e^{-\phi}+\frac{1}{\ell^2} e^{\phi}) dt^2 \right)
\pm e^{2\phi} ( dz  -\frac{\tilde{k}}{2}e^{-2\phi} dt)^2.\end{equation}

Identifying this and the other metrics is particularly easy due to the classification of algebraically special solutions to TMG~\cite{Chow:2009km}. We suspect we are dealing with constant scalar invariant spaces (CSI), after evaluating the first three curvature invariants, in which case they are CSI Kundt, locally homogeneous, or both~\cite{Coley:2007ib,Chow:2009vt}. Furthermore, the Ansatz we use implies two commuting symmetries $\partial_t$ and $\partial_z$. To identify which particular Petrov-Segre class we are in, we study the Jordan normal form of the tensor
$$S_a^{\phantom a b}=R_a^{\phantom a b}-\dfrac{1}{3}R\delta_a^{\phantom a b}.$$
For Case 1, the canonical $S_a^{\phantom a b}$ turns out to be identically zero, \textit{i.e.} the solution is of Petrov class O, corresponding to locally AdS$_3$.

\subsection{Case 2\texorpdfstring{: $c-\delta^2=\tilde{k}=\alpha=0$}{}}
The Ansatz here boils down to
$$f=\mp\dfrac{2\mu}{3}e^{-\phi} + \delta e^{-2\phi},$$
so that $ A=\left(\pm\dfrac{2\mu}{3} e^{-\phi} -\dfrac{1}{2}\delta e^{-2\phi}\right)dt$.
From the equations of motion we obtain
that 
$$\phi'=\mp\dfrac{2}{3}\mu\delta e^{-\phi}+\dfrac{1}{4}\delta^2e^{-2\phi}+d .$$
Furthermore we get that $d=\tfrac{3}{\ell^2}+\tfrac{1}{9}\mu^2$ and
$$\bar R=-\delta^2e^{-2\phi}\pm\dfrac{2}{3}\mu\delta e^{-\phi}.$$
The full 3-dimensional metric is then 
\begin{equation}
\begin{array}{ll}
g_3=&\dfrac{1}{\mp\frac{2}{3}\mu\delta e^{-\phi}+\frac{1}{4}\delta^2e^{-2\phi}+d}d\phi^2 +\left(\pm\frac{2}{3}\mu\delta e^{-\phi}-\frac{1}{4}\delta^2e^{-2\phi}-d\right) dt^2 \\ \\ &+ \left(e^{\phi}dz + \left(\pm\frac{2}{3}\mu -\frac{1}{2}\delta e^{-\phi}\right)dt\right)^2.\end{array}\end{equation}

For this solution, the canonical $S_a^{\phantom a b}$ is given by 
\begin{equation*}S_a^{\phantom a b}=\left(
\begin{array}{ccc}
-\frac{2(\nu^2-1)}{\ell^2} & 0 & 0 \\
0 & \frac{\nu^2-1}{\ell^2} & 0 \\
0& 0& \frac{\nu^2-1}{\ell^2}
\end{array}\right),
\end{equation*}
placing it into Petrov class D, whence by the theorem in~\cite{Chow:2009km} it is locally spacelike or timelike warped AdS$_3$. In fact, case 2 covers both spacelike and timelike stretching. Indeed, one can easily find the diffeomorphism that will bring the metric to one of the standard forms
\begin{equation}  g = \frac{\ell^2}{\nu^2+3}\left( \frac{dy^2}{y^2 - \delta} \mp \left(y^2 - \delta\right)du^2 \pm \frac{4\nu^2}{\nu^2+3}\left( d\tilde{t} +  y du\right)^2\right)~,\end{equation}
where the two values $\delta=0,1$ are isometric, see e.g.\cite{Jugeau:2010nq}. The sign above distinguishes spacelike and timelike stretching and is the same as the one we used to distinguish between spacelike or timelike KK reduction. 

\subsection{Case 3\texorpdfstring{: $\delta=c=\tilde{k}=0$ and $\mu^2\ell^2(2\alpha+{1})^2=(2\alpha-3)^2$}{}}\label{subsec:Case3}
The general Ansatz here is
\[ f=\pm 2\mu \frac{1-2\alpha}{2\alpha-3} \expab{2}{-1}, \]
so that $ A=\mp \frac{2\mu}{2\alpha-3} \expab{2}{-1} dt$.
The equations of motion here yield that 
\[ D^2\phi ={2}\left( \frac{\alpha}{\ell^2}+\mu^2 \frac{({1}-2\alpha)(2\alpha^2+3\alpha)}{(2\alpha-3)^2}\right)e^{2\alpha\phi} ~.\]
Consistency with the dilaton equation requires a $d=0$ integration constant and
\[ \phi'= \frac{4\mu^2}{(2\alpha-3)^2}e^{2\alpha\phi} ~.\]
The 3-dimensional metric is therefore given by
\begin{multline}
g = e^{2\alpha\phi}\left( \frac{d\phi^2}{ \frac{4\mu^2}{(2\alpha-3)^2}e^{2\alpha\phi} } \mp \frac{4\mu^2}{(2\alpha-3)^2}e^{2\alpha\phi} dt^2 \right) \pm e^{2\phi} \left( dz 
\mp \frac{2\mu}{2\alpha-3} \expab{2}{-1} dt \right)^2 .\end{multline}

Again, to identify this solution we look for the Jordan normal form of the traceless Ricci tensor $S_a^{\phantom a b}$, which in this case is
\begin{equation*}S_a^{\phantom a b}=\left(
\begin{array}{ccc}
0& 1 & 0 \\
0 & 0 & 0 \\
0& 0& 0
\end{array}\right),
\end{equation*}
corresponding to the Petrov class N. When $\nu\neq \pm1/3$, a coordinate transformation can bring the metric to the form of an AdS pp-wave
\begin{equation}\label{eq:ourppwave} g = \frac{\ell^2}{4} \frac{d\rho^2}{\rho^2} + s_1 \rho^{\frac{1}{2}(1-3\nu\, s_2)} dz^2 +\rho\, dz \,dt~,\end{equation}
where $s_1$ and $s_2$ are uncorrelated signs. The sign $s_1$ keeps track of the sign of the KK reduction we used and the sign $s_2$ comes from the two possible solutions for $\alpha$. When $\nu=\pm1/3$, the solution for $\alpha$ is unique and our metric becomes that of AdS$_3$ in Poincar\'e coordinates. 

The pp-wave \eqref{eq:ourppwave} then corresponds to a TMG wave \cite{Gibbons:2008vi} with two commuting symmetries. It is locally isometric\footnote{the diffeomorphism in~\cite{Gibbons:2008vi} has an arbitrary function $f_1(z)$ that here should be a constant, see also the appendix in~\cite{Anninos:2010pm}.} to the Schr\"odinger sector solutions of \cite[\S4.2]{Ertl:2010dh} for their $b=0$, which were found and their causal structure analyzed in \cite{Clement:1994sb}. Our Ansatz is thus seen to reproduce locally all known stationary axisymmetric solutions to TMG~\cite{Ertl:2010dh} for generic values $\ell$ and $\nu$, with the excpetion of the $b\neq0$ in \cite[\S4.2]{Ertl:2010dh}. Some more details on this comparison are given in our appendices~\ref{app:mapsol} and~\ref{app:stataxis}.

\section{Conclusion}
The search for new solutions to TMG has lead us to exploring the power and range of the ``kinky'' approach to 3d-gravity as used in \cite{Guralnik:2003we}. The idea of using an exact conformal Killing vector to simplify the reduced two-dimensional equations of motion seems very effective in leading to a whole range of possible solutions depending on a small set of parameters. However, the theorems we used impose strong restrictions on the Ansatz. A subset of valid parameter values is selected that corresponds to the already well-known and studied solutions of locally AdS$_3$, warped AdS$_3$ and the pp-wave. 

Appealing as the Kinky Ansatz may look, it requires too much symmetry to yield any novel solutions. Nonetheless, this is a new, simplified way to obtain the most symmetric TMG backgrounds.  We note how a simple and local Ansatz can reproduce a large class of the known stationary axisymmetric solutions in \cite{Ertl:2010dh}, without assumptions on the asymptotics. In this setting, the relationship between these is in terms of the functional dependence of the generator of a conformal isometry. 

Our Case 3 corresponds to the special case of the family $W_1=-2/\ell$ of type N CSI Kundt solutions where the $f_{01}(u)$ in \cite{Chow:2009vt} is constant. In this way, their general solution acquires an extra isometry, which is precisely what our Ansatz requires. One might wonder whether our Ansatz can be generalized to include other deformations of AdS$_3$, warped AdS$_3$ or the wave - see also appendix~\ref{app:mapsol} for a map of the more general solutions. Another natural question is whether the core idea behind this Ansatz, which was to automatically satisfy the traceless part of the Einstein equation, can be useful in studying other gravitational systems. 

\section*{Acknowledgements}
We would like to thank Gerard Clement for helpful correspondence and Frederic Jugeau for his contribution during the initial stages of this work. P.R. thanks Arjun Bagchi and Jos\'e Figueroa-O'Farrill for useful discussions. 

\appendix

\section{Reduction of Cotton tensor}\label{app:cotton}
For the Kaluza-Klein Ansatz
\begin{equation*}
  g=e^{2\alpha\phi}\bar{g} \pm e^{2\phi}\left(dz+A\right)^2 
\end{equation*}
we can choose an orthonormal basis $\theta^a=e^{\alpha\phi}\bar\theta^a$ and $\theta^z=e^\phi (dz+A)$, where $\bar\theta^a$ is an orthonormal basis of $\bar{g}$. We also define
\begin{align*}
 F &=:  \frac{1}{2} F_{ab} \,\bar{\theta}^a\wedge \bar{\theta}^b \\
 d\phi &=: d\phi_a \, \bar\theta^a~.
\end{align*}
The spin connection of the 3d geometry $\omega_{AB}$ can be solved as
\begin{equation}\label{eq:KKspin}\begin{aligned}
 \omega^{ab} &= \bar{\omega}^{ab} + \alpha (d\phi^b \bar{\theta}^a - d\phi^a \bar{\theta}^b)\mp \frac{1}{2} e^{-2(\alpha-1)\phi}F^{ab}(dz+A)~,
 \\
 \omega^{\bar{z}}{}^a&=e^{-(\alpha-1)\phi}\frac{1}{2} F^a{}_b \bar{\theta}^b +  e^{-(\alpha-1)\phi} d\phi^a (dz+A)~,
\end{aligned}\end{equation}
where $\bar\omega_{ab}$ the spin connection of the 2d geometry.

The information of the curvature two-form
\begin{equation*}
\Omega_{AB}=d\omega_{AB}+\omega_{AC}\wedge \omega^C{}_B =\frac{1}{2}R_{ABCD}\theta^C\wedge\theta^D\end{equation*}
can be encoded for $d=3$ in the Ricci tensor $R_{AB}=R_{ACB}{}^C$. Its components for arbitrary dimension $d$ (in our case $d=3$) are
\begin{equation}\label{eq:KKRicci}\begin{aligned}
 &\begin{aligned}
  R_{ab} &= e^{-2\alpha\phi}\bar{R}_{ab} \mp \frac{1}{2} e^{-4\alpha\phi+2\phi} F_{ac}F_b{}^c 
  - e^{-2\alpha\phi}\left(  \left(\alpha(d-2)+1\right)D_a d\phi_b +\alpha \, 
  D_c d\phi^c \eta_{ab} \right) \\
  &+ e^{-2\alpha\phi} \left(  \left(\alpha^2(d-2)+(2\alpha-1)\right) d\phi_a d\phi_b - \left( \alpha^2(d-2)+\alpha\right)d\phi^c d\phi_c \eta_{ab} \right) 
 \end{aligned}\\
 &
   R_{a}{}^{z} = - e^{-3\alpha\phi+\phi} \frac{1}{2} D_c F^c{}_a + e^{-3\alpha\phi+\phi} \left( 
   -\frac{3}{2} +(2-\frac{d}{2})\alpha
   \right) d\phi^d F_{da} \\
   &
   R_{z}{}^{z}  = - e^{-2\alpha\phi} D_b d\phi^b + e^{-2\alpha\phi}\left( (2\alpha-1)-d\alpha\right)d\phi_c d\phi^c \pm e^{-4\alpha\phi+2\phi} \frac{1}{4} F_{ab} F^{ab}~,
\end{aligned}\end{equation}
where $\bar{R}_{ab}$ is the Ricci tensor of the $d-1$ metric in the basis $\bar\theta^a$.

Using the Ricci tensor \eqref{eq:KKRicci} and the spin connection \eqref{eq:KKspin}, we can find the components of the Cotton tensor
\[ C_{MN} = \epsilon_{M}{}^{PQ} ( E_{NQ;P}+\frac{1}{4} g_{NQ} R_{;P} ) \]
in the orthonormal basis of $\theta^a$. E.g., for a 4d tensor $T_{AB}$ we have
\begin{multline*} E_{az;b} =  e^{-\alpha\phi} D_b E_{az} + \alpha e^{-\alpha\phi} ( \eta_{ab} d\phi^d E_{dz} - d\phi_a E_{bz} ) \\- \frac{1}{2} e^{(-2\alpha+1)\phi} F_{ab} E_{zz} \pm \frac{1}{2} e^{(-2\alpha+1)\phi} E_{ad} F^d{}_b ~,\end{multline*} etc. A useful relation we use in two dimensions to simplify the result is
\[ \epsilon_a{}^c \epsilon_b{}^d T_{(cd)} = \pm ( T_{(ab)} - \eta_{ab} T_c{}^c ) ~.\]
We thus find:
\begin{equation}
\begin{aligned}
\pm \expab{4}{-} C_{ab} &=
\frac{1}{2} (D_a D_b - \eta_{ab} D^2) f - \frac{1}{4} \eta_{ab} f \bar{R} -\frac{1}{2} \eta_{ab} \expab{-2}{+2} f^3 + (1-\alpha) f D_a D_b \phi 
\\
&
+ 4 (\alpha-1)^2 f D_a \phi D_b \phi +\frac{3}{2}(\alpha-1) f \eta_{ab} D^2 \phi - 3(\alpha-1)^2 \eta_{ab} f |d\phi|^2\\
&+\frac{3}{2}(1-\alpha) (D_a \phi D_b f + D_b \phi D_a f) +\frac{5}{2}(\alpha-1)\eta_{ab} D_c\phi D^c f
\end{aligned}
\end{equation}
and
\begin{equation}
C_{az} = \expab{-}{-2} \epsilon_a{}^b D_b \left( \expab{-2}{+2} \left( -\frac{1}{4} \bar{R} - \frac{3}{4} \expab{-2}{+2} f^2 + \frac{1}{2}(\alpha-1)D^2 \phi \right)\right)~.
\end{equation}
The remaining component is $C_{zz}=\mp \eta^{ab} C_{ab}$, which follows because the Cotton tensor is traceless. 

If $E_{MN}$ are the 3d equations of motion, then variation of the reduced action with respect to $g^{ab}$ is proportional to $e^{(2\alpha+1)\phi} E_{ab}$ and variation with respect to $A_\mu$ is proportional to $-2 e^{(\alpha+21)\phi} E_{az}$. What we called the \textbf{F} equation of motion is the first integral of the latter. For the dilaton equation we simply used $R=-\frac{6}{\ell^2}$ and contracted the indices of \eqref{eq:KKRicci}. Alternatively, variation of the reduced action with respect to $\phi$ gives a term proportional to $-2 \,e^{2\alpha\phi} (\pm E_{zz} +\alpha E^a{}_a)$.

Note that, having reduced the theory to $d=2$, the subsequent analysis following section \ref{sec:Ansatz} is made at the level of the equations of motion and not the lagrangian.

\section{Null exact CKV}\label{app:null}
The propositions we use for $d\psi$ non-null in section \ref{sec:theorems} can be altered for a null $d\psi$:
\begin{lemma}
  If $\bar{g}$ has a Conformal Killing one-form $d\psi$ that is {null} and exact, then the metric is flat and for some coordinate system
\begin{equation} \bar{g} =du dv \text{ and }\psi=v~.\end{equation}
\end{lemma}
\begin{lemma}
  Assume two Conformal Killing one-forms, $F_1 d F_2$ and $d\psi$, with $d\psi$ exact and {null} and $F_1$ and $F_2$ functions of $\psi$. They are necessarily related by
\[ F_1 d F_2 = G d\psi, \]
where the $G$ is a function of $v$.
\end{lemma}

The proof of these two lemmas is straightforward and similar to the non-null case. Since $\bar{g}$ is flat, the null case is too restricitve and the Ansatz is not consistent. Notice for instance that the \textbf{Dil.} equation of motion, with $\bar{R}=D^2\phi=|d\phi|^2=0$, implies an imaginary $f$ or $\phi$.

\section{Known solutions to TMG for generic \texorpdfstring{$\nu$}{parameters}}\label{app:mapsol}
The largest set of known solutions is the Kundt class \cite{Chow:2009vt}. The only other known solution, which technically is not a Kundt space, is the timelike warped AdS$_3$ space. The Kundt class of metrics is known in the sense that one can always solve for the metric components given a small set of arbitrary functions that enter their equations of motion. However, there is a subset of Kundt solutions, the ones that are also of Constant Scalar curvature Invariants (CSI), whose solution is known in closed\footnote{more precisely, in terms of integrals of the arbitrary functions.} form~\cite{Chow:2009vt}. Again, the Kundt CSI spacetimes are given in terms of a set of arbitrary functions.

Interestingly, the Kundt CSI spaces can be classified in three classes in terms of their curvature invariants, and for each class they take the same values as one of the three symmetric solutions of section \ref{sec:solutions}: AdS$_3$, (spacelike) warped AdS$_3$ and the pp-wave. That is, any other Kundt CSI solutions can be referred to as a deformation of one of these three spaces. A map of all known solutions is given in figure \ref{fig:mapsol}.

\begin{figure}\begin{center}
\begin{tikzpicture}[grow=right, level distance=3cm]
\node[text width=2cm] {Known solutions} child { node {Kundt} child {node {CSI Kundt} child {node[text width=2.5cm]{stationary axisymmetric Kundt}} }}
child {node {timelike warped AdS$_3$}};
\end{tikzpicture}
\caption{A map of the known solutions to TMG.}\label{fig:mapsol}\end{center}
\end{figure}
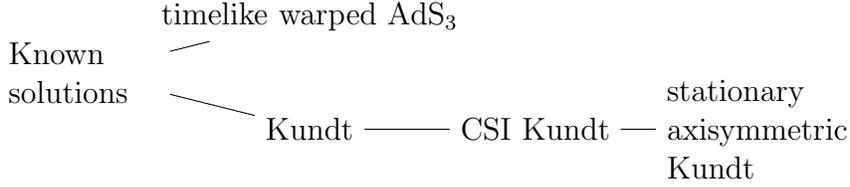

Prior to the work of  \cite{Chow:2009vt}, one large class of Kundt solutions that was known was that of the pp-wave~\cite{Gibbons:2008vi}. The metric can be written in the form
\begin{equation}\label{eq:genpp1} g = d\rho^2 + 2 e^{2\rho} du dv + \left( e^{(1\mp3\nu)\rho} h_1(u) + h_2(u)\right)du^2\end{equation}
with $h_1$ and $h_2$ arbitrary functions. An automorphism of the form \eqref{eq:genpp1} given in~\cite{Gibbons:2008vi} can bring the metric to
\begin{equation}\label{eq:genpp2} g = d\rho^2+ 2 e^{2\rho} du dv + e^{(1-3\nu)\rho} f(u) du^2 ~.\end{equation}
However, in our case 3, such a diffeomorphism is not possible because there $h_1$ and $h_2$ are not arbitrary functions. As a result, the metric of \eqref{eq:ourppwave} is
\begin{equation}\label{eq:stataxispp} g = \frac{\ell^2}{4} \frac{d\rho^2}{\rho^2}+ \left( \pm\rho^{\frac{1}{2}(1\pm3\nu)}+ {{b}}\right) du^2 +\rho\, du \,dv~,\end{equation}
with $b=0$, and is known to not be isometric to \eqref{eq:stataxispp} with $b\neq0$. A different proof of this, based on the dimension of Killing vectors, was given in the appendix of~\cite{Anninos:2010pm}.

The final large class of known solutions are the stationary axisymmetric ones: AdS$_3$ (case 1), spacelike and timelike warepd AdS$_3$ (case 2), and the stationary axisymmetric wave that is precisely \eqref{eq:stataxispp}. For $b=0$ we recover case 3, whereas all $b\neq0$ are among them isometric by a rescaling of the coordinates.

\section{Stationary Axisymmetric solutions}\label{app:stataxis}
A powerful apparatus for stationary axisymmetric solutions given in~\cite{Clement:1994sb} shows that the stationary axisymmetric solutions are equivalent to the motions of a point-like mechanical system. In particular, the two commuting symmetries imply the following metric form
\begin{equation}\label{eq:ClementAnsatz} g= \bar{g}_{\mu\nu}(\rho) dx^\mu dx^\nu + \frac{1}{|\mathbf{X}(\rho)|^2} d\rho^2 \end{equation}
with the unimodular matrix
\[ \bar{g}=\begin{pmatrix} X^+ & Y \\ Y & X^- \end{pmatrix} ~,\]
so that the Lorentz vector $\mathbf{X}=(X^+,X^-,Y)$ is subject to extrema of the lagrangian
\[ S[\mathbf{X}]= \int d\rho \left( \frac{1}{2} \dot{\mathbf{X}}^2-\frac{2}{\ell^2} -\frac{1}{2\mu} \mathbf{X}\cdot (\dot{\mathbf{X}}\times \ddot{\mathbf{X}}) \right) ~.\]

Although stationary axisymmetric solution were initially found in terms of An\-s\"a\-tze for the functional dependence of $\mathbf{X}(\rho)$, they were rederived in \cite{Ertl:2010dh} from first principles. This can be achieved by consistently reducing the phase space by imposing $\ddot{\mathbf{X}}^2=\dot{\mathbf{X}}\cdot\ddot{\mathbf{X}}=0$.  There is only numerical evidence for stationary axisymmetric solutions that do not obey this condition~\cite{Ertl:2010dh}.

It is natural to ask why our Ansatz fails to reproduce \eqref{eq:stataxispp} for $b\neq0$. As it stands, our method is an Ansatz in the true meaning of the term: a test solution that works. However, it does curiously select the undeformed stationary axisymmetric solutions\footnote{Let us here think of the $b\neq 0$ wave as the ``deformation'' of the $b=0$ case. }. The approach presented in this paper can be compared with the work of~\cite{Clement:1990mp}. The general stationary axisymmetric wave (up to identifications) has \cite{Ertl:2010dh}
\begin{equation}\label{eq:sawaveX} \mathbf{X}=( s_1 \rho^{(1\mp 3\nu)/2} + b, 0 ,\pm\frac{2}{\ell}\rho) ~.\end{equation}
Ultimately, the choice $\psi=\phi$ in our setup from section~\ref{sec:ansatz} leads to
\begin{equation}\label{eq:ultmetric} g= e^{2\alpha\phi} \frac{d\phi^2}{\dot{\phi}} \mp\left(e^{2\alpha\phi}\dot{\phi}-A_t(\phi)^2\right)dt^2 \pm e^{2\phi} dz^2\pm 2 e^{2\phi} A_t(\phi) dz dt ~,\end{equation}
where $\dot{\phi}$ a function of $\phi$. Note that \eqref{eq:ClementAnsatz} with \eqref{eq:sawaveX} and $b\neq0$ cannot be written in the form of \eqref{eq:ultmetric} for a change of coordinates $z=f(\rho)$. Indeed, as we saw in section~\ref{subsec:Case3}, consistency for case 3 required $\delta=c=\tilde{k}=0$ and $\mu^2\ell^2(2\alpha+{1})^2=(2\alpha-3)^2$. The parameters $\alpha,\delta,\tilde{k}$ and $c$ are therefore fixed and do not allow deformations that is necessary for a non-zero $b$. This is one way of understanding how (if not why) our Ansatz selects the $b=0$ solution.

\bibliographystyle{utphys}
\bibliography{kinktmgGRG}

\end{document}